\documentclass{stacs_proc}  

\usepackage[latin2]{inputenc}
\usepackage{amsmath,amssymb}
\usepackage{enumerate}
\usepackage[all]{xy}

\newcommand{\limp}{\rightarrow}

\newcommand{\Nat}{\mathbb{N}}

\newcommand{\FO}{\mathsf{FO}}
\newcommand{\FOC}{\mathsf{FOC}}

\newcommand{\dom}{\mathsf{dom}}

\newcommand{\eqe}{\sim_{\textrm{e}}}

\newcommand{\calA}{\mathcal{A}}

\newcommand{\calP}{\mathcal{P}}

\newcommand{\frakd}{\mathfrak{d}}
\newcommand{\frakA}{\mathfrak{A}}

\begin{document}
%%
%% Title
%%

\title[Cardinality and counting quantifiers on omega-automatic structures]
{Cardinality and counting quantifiers on omega-automatic structures}

\author[mgi]{\L{}.~Kaiser}{\L{}ukasz Kaiser}
\author[uoa]{S. Rubin}{Sasha Rubin}
\author[mgi]{V. B\'ar\'any}{Vince  B\'ar\'any}

% \address[mgi]{Mathematische Grundlagen der Informatik \newline RWTH Aachen}
\address[mgi]{Mathematische Grundlagen der Informatik, RWTH Aachen}
\email{{kaiser,vbarany}@informatik.rwth-aachen.de} 

% \address[uoa]{Department of Computer Science \newline University of Auckland}
\address[uoa]{Department of Computer Science, University of Auckland}
\email{rubin@cs.auckland.ac.nz}


\bibliographystyle{plain}

%----------------------------------------------------------------------------
%
% Abstract
%
%----------------------------------------------------------------------------

\begin{abstract}
We investigate structures that can be represented by omega-automata, so called
omega-automatic structures, and prove that relations defined over such
structures in first-order logic expanded by the first-order quantifiers `there
exist at most $\aleph_0$ many', 'there exist finitely many' and 'there exist
$k$ modulo $m$ many' are omega-regular. The proof identifies certain algebraic
properties of omega-semigroups.

As a consequence an omega-regular equivalence relation of countable index has
an omega-regular set of representatives. This implies Blumensath's conjecture
that a countable structure with an $\omega$-automatic presentation can be
represented using automata on finite words. This also complements a very recent
result of Hj\"orth, Khoussainov, Montalban and Nies showing that 
there is an omega-automatic structure which has no injective presentation.

%[Additionally, we identify certain
%algebraic properties of finite omega-semigroups in general and ones that
%represent transitive relations in particular.  These properties might be of
%independent interest and we use them to show a Ramsey-like result for
%$\omega$-regular relations on uncountable structures.]  
\end{abstract}

%In our proof we
%prove that every omega-automatic countable structure has an injective
%presentation, yet leave the general (non-countable case) open.  In fact, an
%example of Kuske and Lohrey shows that injectivity can not generally be achieved 
%by selecting a regular set of unique representatives of every presentation. 

% In the case of presentations using finite words injectivity is easily
% achieved by selecting a regular set of unique represenant, for tree-automatic
% presentations involving finite trees Colcombet and L\"oding gave a
% non-trivial construction.  on the other hand, Kuske and Lohrey gave an
% omega-regular equivalence relation having uncountably many classes and no
% regular set of unique representatives. 

%We show that given a finite automaton recognising an equivalence relation on
%omega-words having only countably many equivalence classes, an omega-regular
%set of representatives can be found effectively.  This implies that every
%omega-automatic countable structure is also automatic, as conjectured by
%Blumensath.  

%----------------------------------------------------------------------------
%
% Title page
%
%----------------------------------------------------------------------------

\maketitle

\stacsheading{2008}{385-396}{Bordeaux}
\firstpageno{385}

%%%%%%%%%%%%%%%%%%%%%%%%%%%%%%%%%%%%%%%%%%%%%%%%%%%%%%%%%%%%%%%%%%%%%%%%%%%%%

\vskip-0.3cm
\section{Introduction} \label{sec_intro}

%%%%%%%%%%%%%%%%%%%%%%%%%%%%%%%%%%%%%%%%%%%%%%%%%%%%%%%%%%%%%%%%%%%%%%%%%%%%%

\setlength{\parskip}{0pt}

% More precisely, for such an $\omega$-regular symmetric relation $R$ we find
% a regular set $L$ so that either every two words from $L$ that are not equal
% almost everywhere are in $R$ or every such two words are not in $R$ 
% (thus forming either an uncountable clique or an uncountable independent set).

Automatic structures were introduced in \cite{Hod83} and later again in
\cite{KN95,BG00} along the lines of the B\"uchi-Rabin equivalence of automata and
monadic second-order logic. The idea is to encode elements of a
structure $\frakA$ via words or labelled trees (the codes need not be unique)
and to represent the relations of $\frakA$ via synchronised automata. This way
we reduce the first-order theory of $\frakA$ to the monadic second-order theory
of one or two successors.  In particular, the encoding of relations defined in
$\frakA$ by first order formulas are also regular, and automata for them can be
computed from the original automata. Thus we have the fundamental fact that the
first-order theory of an automatic structure is decidable. 

Depending on the type of elements encoding the structure, the following natural
classes of structures appear: automatic (finite words), $\omega$-automatic
(infinite words), tree-automatic (finite trees), and $\omega$-tree automatic
(infinite trees). Besides the obvious inclusions, for instance that automatic
structures are also $\omega$-automatic, there are still some some outstanding
problems.  For instance, a presentation over finite words or over finite trees 
can be transformed into one where each element has a unique representative. 

Kuske and Lohrey \cite{KL06} point out an $\omega$-regular
equivalence relation (namely $\eqe$ stating that two infinite words are
position-wise eventually equal) with no $\omega$-regular set of representatives.
Thus, unlike the finite-word case, injectivity can not generally be achieved by
selecting a regular set of representatives from a given presentation. 
In fact, using topological methods it has recently been shown \cite{HKMNman} 
that there are omega-automatic structures having no injective presentation.
However, we are able to prove that every omega-regular equivalence relation
having only countably many classes does allow to select an omega-regular
set of unique representants. Therefore, every countable omega-automatic structure 
does have an injective presentation.

A related question raised by Blumensath \cite{Blu99} is whether every countable
$\omega$-automatic structure is also automatic. In Corollary~\ref{coroll_cnt} 
we confirm this by transforming the given presentation into an injective one, 
and then noting that an injective $\omega$-automatic presentation of a countable 
structure can be ``packed'' into one over finite words.

All these results rest on our main contribution: a characterisation of when
there exist countably many words $x$ satisfying a given formula with parameters
in a given $\omega$-automatic structure $\frakA$ (with no restriction on the cardinality
of the domain of $\frakA$ or the injectivity of the presentation). The
characterisation is first-order expressible in an $\omega$-automatic
presentation of an extension of $\frakA$ by $\eqe$.
Hence we obtain an extension of the fundamental fact for $\omega$-automatic
structures to include cardinality and counting quantifiers such as 'there exists
(un)countably many', 'there exists finitely many', and 'there exists $k$ modulo
$m$ many'.  This generalises results of Kuske and Lohrey \cite{KL06} who achieve
this for structures with \emph{injective} $\omega$-automatic presentations.
 
%%%%%%%%%%%%%%%%%%%%%%%%%%%%%%%%%%%%%%%%%%%%%%%%%%%%%%%%%%%%%%%%%%%%%%%%%%%%%

\vskip-0.3cm
\section{Preliminaries} \label{sec_prelim}

%%%%%%%%%%%%%%%%%%%%%%%%%%%%%%%%%%%%%%%%%%%%%%%%%%%%%%%%%%%%%%%%%%%%%%%%%%%%%

By countable we mean finite or countably infinite. 
Let $\Sigma$ be a finite alphabet. With $\Sigma^\ast$ and $\Sigma^\omega$
we denote the set of finite, respectively $\omega$-words over $\Sigma$.
The length of a word $w \in \Sigma^\ast$ is denoted by $|w|$,
the empty word by $\varepsilon$, 
and for each $0 \leq i < |w|$ the $i$th symbol of $w$ is written as $w[i]$. 
Similarly $w[n,m]$ is the factor $w[n]w[n+1]\cdots w[m]$ and $w[n,m)$ is defined
by $w[n,m-1]$.
Note that we start indexing with $0$ and that for $u \in \Sigma^\ast$ we
denote by $u^n$ the concatenation of $n$ number of $u$s, in particular
$u^\omega \in \Sigma^\omega$. 

We consider relations on finite and $\omega$-words recognised by multi-tape
finite automata operating in a synchronised letter-to-letter fashion.
Formally, \emph{$\omega$-regular relations} are those accepted by some
finite non-deterministic automaton $\calA$ with B\"uchi, parity or Muller
acceptance conditions, collectively known as $\omega$-automata,
% (in the latter two cases determinism can be assumed)
and having transitions labelled by $m$-tuples of symbols of $\Sigma$.
Equivalently, $\calA$ is a usual one-tape $\omega$-automaton over the
alphabet $\Sigma^m$ accepting the \emph{convolution} $\otimes\vec{w}$ of
$\omega$-words $w_1,\ldots,w_m$ defined by 
$\otimes\vec{w}[i] = (w_1[i], \ldots, w_m[i])$ for all $i$.
% For pairs of words $w, w'$ we will sometimes denote $w \otimes w'$
% by $\binom{w}{w'}$. %% We do not use it except for letters in the picture.

Words $u,v \in \Sigma^{\omega}$ have {\em equal ends}, written $u \eqe v$, if
for almost all $n \in \Nat$, $u[n] = v[n]$.  This is an important
$\omega$-regular equivalence relation. We overload notation so that for $S,T
\subset \Nat$ we write $S \eqe T$ to mean for almost all $n \in \Nat$, $n \in S
\iff n \in T$.
 
\begin{example}\label{ex_1} The non-deterministic B\"uchi automaton depicted in
Fig. \ref{fig_ex1} accepts the equal-ends relation on alphabet $\{0,1\}$. 

\begin{figure}[h]
\begin{displaymath}
\xymatrix{
   1 \ar@(ul,ur)^{\binom{0}{0},\binom{0}{1},\binom{1}{0},\binom{1}{1}} 
     \ar [rr]^{\binom{0}{0},\binom{1}{1}} & *{} & 
   2 \ar@(ul,ur)^{\binom{0}{0},\binom{1}{1}} &  F = \{2\}, I = \{1\} \\
}
\end{displaymath}
\caption{An automaton for the equal ends relation $\eqe$.}
\label{fig_ex1}
\end{figure}
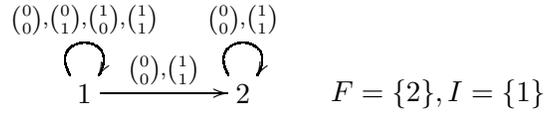
\end{example}

In the case of finite words one needs to introduce a padding
end-of-word symbol $\Box \not\in \Sigma$ 
to formally define convolution of words of different length. 
For simplicity, we shall identify each finite word $w \in \Sigma^\ast$
with its infinite padding $w^{\Box} = w \Box^\omega \in \Sigma_\Box^\omega$ 
where $\Sigma_\Box = \Sigma \cup \{\Box\}$. 
To avoid repeating the definition of automata for finite words,
we say that a $m$-ary relation $R \subseteq (\Sigma^\ast)^m$ 
is \emph{regular} (\emph{synchronised rational}) whenever 
it is $\omega$-regular over $\Sigma_\Box$.

\vskip-0.3cm
\subsection{Automatic structures}

%Having defined ($\omega$-)regular relations we say that a finite automaton $\calA$ 
%accepting the ($\omega$-)regular relation $R$ represents this relation.
%Accordingly, an ($\omega$-)automatic presentation of a relational structure is
%defined as follows.
We now define what it means for a relational structure (we implicitly replace
any structure with its relational counterpart) to have an ($\omega$-)automatic
presentation. 

\begin{definition}[($\omega$-)Automatic presentations] \label{def_as} ~\\
Consider a relational structure $\frakA = (A, \{R_i\}_i)$ with universe
$\dom(\frakA) = A$ and relations $R_i$.
A tuple of $\omega$-automata 
$\frakd = (\calA, \calA_\approx, \{\calA_i\}_i)$
together with a surjective naming function $f : L(\calA) \to A$
constitutes an \emph{($\omega$-)automatic presentation} of $\calA$
if the following criteria are met: 
\begin{enumerate}[(i)]
\item the equivalence, denoted $\approx$, and defined by $\{ (u,w) \in L(\calA)^2 \mid f(u) = f(w) \}$ 
   is recognised by $\calA_\approx$,
\item every $L(\calA_i)$ has the same arity as $R_i$,
\item $f$ is an isomorphism between 
   $\frakA_\frakd = (L(\calA), \{L(\calA_i)\}_i) /_{\approx}$ and $\frakA$.
\end{enumerate}
The presentation is said to be \emph{injective} whenever $f$ is, 
in which case $\calA_\approx$ can be omitted.
\end{definition}

The relation $\approx$ needs to be a congruence of the structure 
$(L(\calA), \{L(\calA_i)\}_i)$ for item (iii) to make sense. 
In case $L(\calA)$ only consists of words of the form $w^\Box$ where 
$w \in \Sigma^\ast$, we say that the presentation is {\em automatic}.  
Call a structure {\em $(\omega-)$automatic} if it has an ($\omega$-)automatic 
presentation.

% \begin{example} \label{ex_complete}
% Let $\Sigma$ be a finite alphabet.
% Let $S_a$, $\pref$ and $\el$ denote the $a$-successor relation, the prefix ordering, 
% and the relation consisting of pairs of words of equal length. These relations are 
% clearly regular, thus $\frakS_\Sigma  = (\Sigma^*,(S_a)_{a\in\Sigma},\preceq,\el)$ is 
% automatic, having $\mathsf{id} \in \AP(\frakS_\Sigma)$.
% Note that $\frakS_{\{1\}}$ is essentially $(\Nat,\leq)$.
% \end{example}

% Moreover, for each non-unary $\Sigma$, $\frakS_\Sigma$ is \emph{complete}, 
% in the sense of item ii) below, for $\AS$ with respect to first-order interpretations.

% \begin{proposition} \label{prop_complete} \cite{BG04}  
% Let $\Sigma$ be a finite, \emph{non-unary} alphabet. \\[-1.7 em]
% \begin{enumerate}
% \item[i)]  A relation $R$ over $\Sigma^*$ is regular if and only if 
%             it is definable in $\frakS_\Sigma$.
% \item[ii)] A structure $\frakA$ is automatic if and only if it is first-order 
%             interpretable in $\frakS_\Sigma$. 
% \end{enumerate}
% \end{proposition}

The advantage of having an ($\omega$-)automatic presentation of a structure 
lies in the fact that first-order ($\FO$) formulas can be effectively evaluated
using classical automata constructions. This is expressed by the following 
fundamental theorem.

\begin{theorem} \label{thrm_fo}
   (Cf. \cite{Hod83}, \cite{KN95}, \cite{BG04}.) ~\\[-1 em]
\begin{enumerate}[(i)]
\item There is an effective procedure that given an $(\omega-)$automatic
presentation $\frakd,f$ of a structure $\frakA$, and given a $\FO$-formula
$\varphi(\vec{a},\vec{x})$ with parameters $\vec{a}$ from $\frakA$ (defining a
$k$-ary relation $R$ over $\frakA$), constructs a $k$-tape synchronous
$(\omega-)$automaton recognising $f^{-1}(R)$.

\item The $\FO$-theory of every $(\omega-)$automatic structure is decidable. 
\item The class of $(\omega-)$automatic structures is closed under 
  $\FO$-interpretations 
\end{enumerate}
\end{theorem}

%[Observe that by $(i)$ above every set and relation first-order definable 
%using only ($\omega$-)regular sets and relations is again ($\omega$-)regular. 
%We shall use this  fact without explicit reference.]

Let $\FOC$ denote the extension of first-order logic with all quantifiers 
of the form 
\begin{itemize}
\item $\exists^{(r\,\mathrm{mod}\,m)} x \, . \, \varphi$ 
      meaning that the number of $x$ satisfying $\varphi$ is finite 
      and is congruent to $r$ mod $m$; 
\item $\exists^\infty x \, . \, \varphi$ 
      meaning that there are infinitely many $x$ satisfying $\varphi$;
\item $\exists^{\leq \aleph_0} x \, . \, \varphi$ and $\exists^{>\aleph_0} x  \, . \, \varphi$
      meaning that the cardinality of the set of all $x$ satisfying $\varphi$
      is countable, or uncountable, respectively.
\end{itemize}
It has been observed that for \emph{injective} ($\omega$-)automatic 
presentations Theorem \ref{thrm_fo} can be extended from $\FO$ to $\FOC$ \cite{KRS04,KL06}.
Moreover, Kuske and Lohrey show that the cardinality of any set definable 
in $\FOC$ is either countable or equal to that of the continuum.
Our main contribution is the following generalisation of their result.

\begin{theorem} \label{thrm_foc}
The statements of Theorem \ref{thrm_fo} hold true for $\FOC$ over all
(not necessarily injective) $\omega$-automatic presentations.
\end{theorem}

It is easily seen that finite-word automatic presentations can be assumed to be
injective. 
% and over the binary alphabet $\{0,1\}$.  
This is achieved by restricting the domain of the presentation to a regular set
of representatives of the equivalence involved.  This can be done effectively,
e.g. by selecting the length-lexicographically least word of every class.

This brings us to the question which $\omega$-automatic structures allow 
an injective $\omega$-automatic presentation.
In \cite{KL06} Kuske and Lohrey have pointed out that not every
$\omega$-regular equivalence has an $\omega$-regular set of representatives.
In particular, the following Lemma shows that the \emph{equal-ends relation
$\eqe$} of Example \ref{ex_1} is a counterexample.

\begin{lemma}[{\cite[Lemma 2.4]{KL06}}] \label{lemma_KL}
Let $\calA$ be a B\"uchi automaton
with $n$ states over $\Sigma \times \Gamma$ and let $u \in \Sigma^\omega$ be given.
Consider the set $V = \{ v \in \Gamma^\omega \mid u \otimes v \in L(\calA) \}$.  
Then $V$ is uncountable if and only if $|V / \eqe | > n$, otherwise it is finite
or countable.
\end{lemma}

The lemma implies that an $\omega$-regular set is countable if and only if it
meets only finitely many equal-ends-classes.  In this case each of its members
is ultimately periodic with one of finitely many periods. 

% \footenote{This is somewhat analogous to the 
% characterisation of regular languages with polynomial growth due to 
% Szilard et al. \cite{SzYZhSh92}. Not surprisingly for, as it is easy to check,
% a \emph{closed} $\omega$-regular set $L$ is finite, countably infinite, or 
% uncountable iff 
% the (regular) set of its finite prefixes is of constant, polynomial, 
% respectively, of exponential density.}

\begin{corollary} \label{coroll_countable}
An $\omega$-regular set is countable iff it
can be written as a finite union of sets of the form $U_j \cdot (w_j)^\omega$ 
with each $U_j$ a regular set of finite words and each $w_j$ a finite word.
% It can moreover be assumed that we have a disjoint union of the above form
% with pairwise non-conjugate primitive words $w_i$.
\end{corollary}

% \begin{proof}
% Every set of the given form is clearly $\omega$-regular and countable or finite.
% Conversely, we know that every $\omega$-regular set has an ultimately periodic 
% element.
% The words of $L$ in the same equal-ends class (having the same period) form a 
% set as required. Removing them from $L$ we proceed until no more words are found. 
% By assumption this will happen after a finite number of steps.
% \end{proof}

A related question raised by Blumensath \cite{Blu99} is whether every \emph{countable} 
$\omega$-automatic structure is also automatic. It is easy to see
that every \emph{injective} $\omega$-automatic presentation of a countable
structure can be ``packed'' into an automatic presentation. 
%We sketch a proof of this fact. % for the reader's convenience 

\begin{proposition}\label{prop_cnt_inj} (\cite[Theorem 5.32]{Blu99})
Let $\frakd$ be an injective $\omega$-automatic presentation of 
a countable structure $\calA$.
Then, an (injective) automatic presentation $\frakd'$ of $\calA$ 
can be effectively constructed.
\end{proposition}

%\begin{proof}
%By injectivity of the presentation, its domain $D$ is a countable $\omega$-regular
%set and therefore of the form $\bigcup_{k=1}^{n} U_k (w_k)^\omega$ for some finite 
%periods $w_k$ and regular sets $U_k \subseteq \Sigma^\ast$.
%We may assume that $\{1,\ldots,n\} \cap \Sigma = \emptyset$.  To obtain a
%presentation using finite words let $D' = \bigcup_k{k \cdot U_k}$ be the domain
%and let the word $k\cdot u$ code the element represented by $u(w_k)^{\omega}$
%in the original presentation. Each induced relation is recognised by an
%automaton $\calA_i'$ that can be constructed from the $\omega$-automaton
%$\calA_i$. 
%On reading, as the first symbols of its input, a tuple $(k_1,\ldots,k_{n_i})$ 
%of indices, automaton $\calA_i'$ enters a state $(q_0,w_{k_1},\ldots,w_{k_{n_i}})$
%where $q_0$ is the initial state of $\calA_i$ and the $w_{k_j}$'s are the 
%periods corresponding to the indices $k_j$.  $\calA_i'$ simulates $\calA_i$
%and, on reaching the end of certain input words, $\calA_i'$ proceeds by
%treating the corresponding period $w_{k_j}$ stored in its state as part of the
%input, and rotating it in each step by one letter.  Upon termination, the state
%of $\calA_i'$, say $(q,\widetilde{w_{k_1}},\ldots,\widetilde{w_{k_{n_i}}})$,
%is accepting precisely if $\calA_i$ accepts the tuple
%$(\widetilde{w_{k_1}}^\omega,\ldots,\widetilde{w_{k_{n_i}}}^\omega)$ from state
%$q$.  \end{proof}

%Note that in the proof above we have only relied on the countability
%of the domain of the $\omega$-automatic presentation, not on it
%actually being injective.

In our proof of Theorem \ref{thrm_foc} we identify a property of finite
semigroups that recognise transitive relations (Lemma~\ref{lemma_ab} item (3))
that allows us to drop the assumption of injectivity in the previous statement.
We are thus able to answer the question of Blumensath.

\begin{corollary}\label{coroll_cnt}
A countable structure is $\omega$-automatic if and only if it is automatic.
Transforming a presentation of one type into the other can be done effectively.
\end{corollary}

\vskip-0.3cm
\subsection{$\omega$-Semigroups}

The fundamental correspondence between recognisability by
finite automata and by finite semigroups has been extended to 
$\omega$-regular sets. 
This is based on the notion of \emph{$\omega$-semigroups}.
Rudimentary facts on $\omega$-semigroups are well presented in \cite{PP95}.
We only mention what is most necessary.

An $\omega$-semigroup $S = (S_f,S_\omega,\cdot,\ast, \pi)$ 
is a two-sorted algebra, where $(S_f,\cdot)$ is a semigroup, 
$\ast : S_f \times S_\omega \mapsto S_\omega$ is the \emph{mixed product}
satisfying for every $s,t \in S_f$ and every $\alpha \in S_\omega$ the equality
\[ 
  s \cdot (t \ast \alpha) = (s \cdot t) \ast \alpha 
\]
and where $\pi: S_f^\omega \mapsto S_\omega$ is the \emph{infinite product}
satisfying 
\[
  s_0 \cdot \pi(s_1, s_2, \ldots) = \pi(s_0, s_1, s_2, \ldots)
\]
as well as the associativity rule
\[ 
  \pi(s_0, s_1, s_2, \ldots) =  
    \pi(s_0 s_1 \cdots s_{k_1}, s_{k_1 +1} s_{k_1 +2} \cdots s_{k_2}, \ldots)
\]
for every sequence $(s_i)_{i\geq 0}$ of elements of $S_f$ and every 
strictly increasing sequence $(k_i)_{i\geq 0}$ of indices.
For $s \in S_f$ we denote $s^\omega = \pi(s, s, \ldots)$.

%Give definition of morphism?
Morphisms of $\omega$-semigroups are defined to preserve all three products
as expected. There is a natural way to extend finite semigroups and their 
morphisms to $\omega$-semigroups. As in semigroup theory, idempotents play 
a central role in this extension.
An \emph{idempotent} is a semigroup element $e \in S$ satisfying $ee = e$. 
For every element $s$ in a finite semigroup the sub-semigroup generated by $s$ 
contains a unique idempotent $s^k$. The least $k>0$ such that $s^k$ is idempotent 
for every $s \in S_f$ is called the \emph{exponent} of the semigroup $S_f$ 
and is denoted by $\pi$. Another useful notion is absorption of semigroup elements: 
say that {\em $s$ absorbs $t$ (on the right)} if $st = s$. 

%We will use the following simple fact: for all $s,t$ the idempotent element
%$(s^k t^k)^k$ absorbs the idempotent element $(t^k t^k)^k$.
%Another key notion is that of \emph{linked pairs} 
%$(s,e)$ with $e$ an idempotent and such that $se=s$. 
%A way to think of a linked pair is as of an initial path labelled $s$ leading 
%into a loop labelled $e$ in a finite graph, only ``on a global scale'', that is, 
%starting in any state.

There is also a natural extension of the free semigroup $\Sigma^+$ to 
the $\omega$-semigroup $(\Sigma^+,\Sigma^\omega)$ with $\ast$ and $\pi$ 
determined by concatenation. 
An $\omega$-semigroup $S=(S_f, S_\omega)$ \emph{recognises} a language 
$L \subseteq \Sigma^\omega$ via a morphism 
$\phi : (\Sigma^+,\Sigma^\omega) \rightarrow (S_f, S_\omega)$ 
if $\phi^{-1}(\phi(L)) = L$.
This notion of recognisability coincides, as for finite words,
with that by non-deterministic B\"uchi automata. 
In \cite{PP95} constructions from B\"uchi automata to $\omega$-semigroups
and back are also presented.

\begin{theorem}[\cite{PP95}] \label{thrm_omega_semigroups} ~\\
A language $L \subseteq \Sigma^\omega$ is $\omega$-regular iff
it is recognised by a finite $\omega$-semigroup.
\end{theorem}

We note that this correspondence allows one to engage 
in an algebraic study of varieties of $\omega$-regular languages,
and also has the advantage of hiding complications of cutting apart and
stitching together runs of B\"uchi automata as we shall do.
This is precisely the reason that we use this algebraic framework.
Most remarkably, one does not need to understand the exact relationship between
automata and $\omega$-semigroups and the technical details of the constructions
behind Theorem \ref{thrm_omega_semigroups} to comprehend our proof.
An alternative approach, though likely less advantageous, would be to use
the composition method, which is closer in spirit to $\omega$-semigroups
than to automata. \footnote{
Define $T_f$ resp. $T_\omega$ as the sets of bounded (in terms of quantifier rank) 
theories of finite, respectively, of $\omega$-words. The composition theorem 
ensures that $\cdot$, $\ast$, $\pi$ can naturally be defined on bounded theories.}

%%%%%%%%%%%%%%%%%%%%%%%%%%%%%%%%%%%%%%%%%%%%%%%%%%%%%%%%%%%%%%%%%%%%%%%%%%%%%
%%
%% Cardinality and modulo counting quantifiers
%%
%%%%%%%%%%%%%%%%%%%%%%%%%%%%%%%%%%%%%%%%%%%%%%%%%%%%%%%%%%%%%%%%%%%%%%%%%%%%%

\vskip-0.3cm
\section{Cardinality and modulo counting quantifiers}

This section is devoted to establishing the key to Theorem \ref{thrm_foc} announced earlier.

%{\bf Theorem \ref{thrm_foc}}
%{\it The statements of Theorem \ref{thrm_fo} hold true for $\FOC$ 
%     over all (not necessarily injective) $\omega$-automatic presentations.}

We characterise when there exist countably many words $x$ satisfying a given
formula {\em with parameters} $\varphi(x,\vec{z})$ in some $\omega$-automatic
structure $\frakA$.  The characterisation is first-order expressible in an
$\omega$-automatic extension of $\frakA$ by the equal-ends relation $\eqe$. 

%[Moreover, the quantifier rank of the resulting
%formula depends on the constant $C$, which itself depends the given
%presentation of $\frakA$.]

So, fix an $\omega$-automatic presentation 
%$\frakd = (D, \approx, \ldots)$ 
of some $\frakA$ with congruence $\approx$, and a first-order formula
$\varphi(x,\vec{z})$ in the language of $\frakA$ with $x$ and $\vec{z}$ free
variables.

\begin{proposition} \label{prop_count}
There is a constant $C$, computable from the presentation $\frakd$, so that
for all tuples $\vec{z}$ of infinite words the following are equivalent:
\begin{enumerate}

\item $\varphi(-,\vec{z})$ is satisfiable and $\approx$ restricted to the 
      domain $\varphi(-,\vec{z})$ has countably many equivalence classes.

\item there exist $C$-many words $x_1, \cdots, x_C$ each satisfying
      $\varphi(-,\vec{z})$, so that every $x$ satisfying
      $\varphi(-,\vec{z})$ is $\approx$-equivalent to some $y \eqe x_i$. 
      Formally, the structure $(\frakA, \approx, \eqe)$ models the sentence below. 
\[
  \forall \vec{z} \left(
    \exists^{\leq \aleph_0} w \, . \, \varphi(w,\vec{z})  
  \longleftrightarrow  
    \exists x_1 \ldots x_C \left( 
      \bigwedge_i \varphi(x_i,\vec{z}) \land 
      \forall x \varphi(x,\vec{z}) \limp 
        \exists y (x \approx y \land \bigvee_i y \eqe x_i) 
    \right)
  \right)
\]
\end{enumerate}
\end{proposition}

\begin{proof}
Suppose $\frakd$, $\frakA$, and $\varphi$ are given.  Define $C$ to be $c^2$,
where $c$ is the size of the largest $\omega$-semigroup corresponding to any of
the given automata (from the presentation or corresponding to $\varphi$).  Now
fix parameters $\vec{z}$. From now on, $\approx$ denotes the equivalence
relation $\approx$ restricted to domain $\varphi(-,\vec{z})$.

$2 \rightarrow 1$: Condition $2$ and the fact that every $\eqe$-class is
countable imply that all words satisfying $\varphi(-,\vec{z})$ are contained  
in a countable number of $\approx$-classes.

$1 \rightarrow 2$: We prove the contra-positive in three steps. 

% The negation of condition $2$ says that given $D < C$ many words $x_1, \cdots,
% x_D$ (each satisfying $\varphi(-,\vec{z})$) there exists a word $x_{D+1}$ also
% satisfying $\varphi(-,\vec{z})$ whose $\approx$-class does not meet any of the
% $\eqe$-classes of the $x_i$ for $i \leq D$.
%
% So inductively define words $x_1,\cdots,x_C$ each satisfying
% $\varphi(-\vec{z})$, and so that for $1 \leq i < j \leq C$ the $\approx$-class
% of $x_j$ does not meet the $\eqe$-class of $x_i$. In particular, the $x_i$s are
% pairwise~$\not\eqe$. 

If $\varphi(-,\vec{z})$ is satisfiable then the negation of condition $2$ 
implies that there are $C+1$ many words $x_0, \ldots, x_C$ each satisfying 
$\varphi(-,\vec{z})$, and so that for $i,j \leq C$, $i \neq j$, the $\approx$-class 
of $x_j$ does not meet the $\eqe$-class of $x_i$. In particular, the $x_i$s are 
pairwise~$\not\eqe$. 

The plan is to produce uncountably many pairwise non-$\approx$ words that
satisfy $\varphi(-,\vec{z})$. In the first 'Ramsey step', similar to
what is done in \cite{KL06}, we find two words from the given $C$ many,
say $x_1,x_2 \in \Sigma^\ast$, and a factorisation $H \subset \Nat$ so that
both words behave the same way along the factored sub-words with
respect to the $\approx$- and $\varphi$-semigroups. 
In the second 'Coarsening step' we identify a technical property of finite
semigroups recognising transitive relations. This allows us to produce
an altered factorisation $G$ and new, well-behaving words $y_1,y_2$. 
In the final step, the new words are 'shuffled along $G$'
to produce continuum many pairwise non-$\approx$ words, each satisfying
$\varphi(-,\vec{z})$.

\vskip-0.3cm
\subsection{Ramsey step}

This step effectively allows us to discard the parameters $\vec{z}$.
Before we use Ramsey's theorem, we introduce a convenient notation 
to talk about factorisations of words.

\begin{definition}
Let $A = a_1 < a_2 < \cdots$ be any subset of $\Nat$ and $h:\Sigma^\ast \to S$ 
be a morphism into a finite semigroup $S$.
For an $\omega$-word $\alpha \in \Sigma^\omega$, and element $e \in S$, 
say that {\em $A$ is an $h,e$-homogeneous factorisation of $\alpha$} if
for all $n \in \Nat^+$, $h\big(\,\alpha[a_n,a_{n+1})\,\big) = e$.  
\end{definition}

\noindent Observe that
\begin{enumerate}
\item if $A$ is an $h,s$-homogeneous factorisation of $\alpha$ and $k \in \Nat^+$ 
      then the set $\{a_{ki}\}_{i \in \Nat^+}$ is an $h,{s^k}$-homogeneous factorisation 
      of $\alpha$.
\item if $A$ is an $h,e$-homogeneous factorisation of $\alpha$ and $e$ is idempotent, 
      then every infinite $B \subset A$ is also an $h,e$-homogeneous factorisation 
      of $\alpha$.
\end{enumerate}

%If furthermore $w[0,a_1) \cdot S$ absorbs $e$, then say that $A$ is {\em
%$S,e$-strongly homogeneous}.
%write $z|A,n$ for $z[a_n,a_{n+1}]$.  
%When $A$ is understood, write $z|n$.
\noindent 
In the following we write $w^\varphi$ and $w^\approx$ to denote the image 
of $w$ under the semigroup morphism into the finite semigroup associated 
to $\varphi$ and $\approx$, respectively, as determined by the presentation. 
Accordingly, we will speak of e.g. $\varphi,s_i$-homogeneous factorisations.
%For semigroup $S \in \{\varphi, \approx\}$
%and any infinite set $A \subset \Nat$ we write $w^S$ for $\phi_S(w)$
%where $w \in \Sigma^+$.

Let us now colour every $\{n,m\} \in [\Nat]^2$, say $n < m$, by the tuple 
of $\omega$-semigroup elements
\[
\langle \
\big( \otimes(x_i,\vec{z})[n,m)^\varphi \, \big)_{0 \leq i \leq C}\ , \
\big( \otimes(x_i,x_j)[n,m)^\approx \, \big)_{0 \leq i \leq j \leq C} \
\rangle.
\]
By Ramsey's theorem there exists infinite $H \subset \Nat$ and a tuple of
$\omega$-semigroup elements 
\[
\left<
(s_i)_{1 \leq i \leq C}, 
(t_{(i,j)})_{1 \leq i \leq j \leq C}
\right>
\]
so that for all $0 \leq i \leq j \leq C$,
\begin{itemize}
\item $H$ is a $\varphi,s_i$-homogeneous factorisation of the word $\otimes(x_i,\vec{z})$, 
\item $H$ is $\approx,t_{(i,j)}$-homogeneous factorisation of the word $\otimes(x_i,x_j)$. 
\end{itemize}

Note that by virtue of additivity of our colouring and Ramsey's theorem 
each of the $s_i$ and $t_{(i,j)}$ above are idempotents.
Note that since there are at most $c$-many $s_i$s and $c$-many $t_{(i,i)}$s
there are at most $c^2$ many pairs $(s_i,t_{(i,i)})$ and so there must be two
indices, we may suppose $1$ and $2$, with $s_1 = s_2$ and $t_{(1,1)} =
t_{(2,2)}$. 

%In other words, \dots

\vskip-0.3cm
\subsection{Coarsening step}

For technical reasons we now refine $H$ and alter $x_1, x_2$ so that the
semigroup elements have certain additional properties.

To start with, using the fact that $x_1 \not\eqe x_2$ and our observation
on coarsenings, we assume without loss of generality that $H$ is coarse enough
so that $x_1[h_n,h_{n+1}) \neq x_2[h_n,h_{n+1})$ for all $n \in \Nat$.
% \begin{enumerate}
% \item $s_i$ and $t_{(i,j)}$ ($i,j \in \{1,2\}$) are idempotent, and
% \item $x_1[h_n,h_{n+1}) \neq x_2[h_n,h_{n+1})$ for all $n \in \Nat$.
% \end{enumerate}

\begin{lemma} \label{lemma_ab}
There exists a subset $G \subset H$, listed as $g_1 < g_2 <
\cdots$, and $\omega$-words $y_1,y_2$ with the following properties:
\begin{enumerate}

\item The words $y_1$ and $y_2$ are neither $\approx$-equivalent 
      nor $\eqe$-equivalent, and each satisfies $\varphi(-,\vec{z})$.

\item There exists an idempotent $\varphi$-semigroup element $s$ 
      such that $G$ is a $\varphi,s$-homogeneous factorisation 
      for each of $\otimes(y_1,\vec{z})$ and $\otimes(y_2,\vec{z})$. 

\item There exist idempotent $\approx$-semigroup elements
      $t,t^\uparrow,t^\downarrow$ so that for $y_j \in \{y_1,y_2\}$ 
\begin{itemize}
\item both $t^\uparrow$ and $t^\downarrow$ absorb $t$
\item $\otimes(y_j,y_j)[0,g_1)^\approx$ absorbs $t$
\item $G$ is an $\approx,t$-homogeneous factorisation of $\otimes(y_j,y_j)$
\item $G$ is an $\approx,t^\uparrow$-homogeneous factorisation of $\otimes(y_1,y_2)$
\item $G$ is an $\approx,t^\downarrow$-homogeneous factorisation of $\otimes(y_2,y_1)$.
\end{itemize}
\end{enumerate}
\end{lemma}

\begin{proof}
Define $\omega$-words 
$y_1:=x_2[0,h_2)x_1[h_2,\infty)$, and $y_2$ by
\begin{eqnarray*}
y_2[0,h_2) & := & x_2[0,h_2) \text{ and}\\
y_2[h_{2n},h_{2n+2}) & := & x_2[h_{2n},h_{2n+1}) x_1[h_{2n+1},h_{2n+2}) \text{ for } n > 0.
\end{eqnarray*}

\noindent
{\em Item 1.} Clearly, $y_1 \not\eqe y_2$ and
each $y_j \in \{y_1,y_2\}$ satisfies $\varphi(y_j,\vec{z})$ since by homogeneity 
and $s_1 = s_2$
\begin{eqnarray*}
\otimes(y_1,\vec{z})^\varphi 	& = &	\otimes(x_2,\vec{z})[0,h_2)^\varphi s_1^\omega \\
				& = &	\otimes(x_2,\vec{z})[0,h_2)^\varphi s_2^\omega \\
				& = &	\otimes(x_2,\vec{z})^\varphi
\end{eqnarray*}
and similarly
\begin{eqnarray*}
\otimes(y_2,\vec{z})^\varphi 	& = &	\otimes(x_2,\vec{z})[0,h_2)^\varphi (s_2 s_1)^\omega \\
				& = &	\otimes(x_2,\vec{z})[0,h_2)^\varphi s_2^\omega \\
				& = &	\otimes(x_2,\vec{z})^\varphi
\end{eqnarray*}

Next we check that $y_1 \not \approx y_2$. 
%Write $e$ for $\otimes(x_2,x_2)[0,h_2)^\approx$.
\begin{eqnarray*}
\otimes(y_1,y_2)^\approx	
%--------------------------------------------------------------------------------
	& = & 	\pi_\approx \big( 
		\otimes(x_2,x_2)[0,h_2)^\approx,\, 
		\big( \otimes(x_1,x_2)[h_{2n},h_{2n+1})^\approx,\, 
		 \otimes(x_1,x_1)[h_{2n+1},h_{2n+2})^\approx
		\big)_{n  \in \Nat^+}
		\big)\\
%--------------------------------------------------------------------------------
	& = & 	\otimes(x_2,x_2)[0,h_1)^\approx \, t_{(2,2)} \ 
		(t_{(1,2)}t_{(1,1)})^{\omega}\\
%--------------------------------------------------------------------------------
	& = & 	\otimes(x_2,x_2)[0,h_1)^\approx \,  
		t_{(2,2)}t_{(2,2)} \  
		(t_{(1,2)}t_{(1,1)})^{\omega}\\
%--------------------------------------------------------------------------------
	& = & 	\otimes(x_2,x_2)[0,h_1)^\approx \,  
		t_{(2,2)}t_{(2,2)} \  
		(t_{(1,2)}t_{(2,2)})^{\omega}\\
%--------------------------------------------------------------------------------
	& = & 	\otimes(x_2,x_2)[0,h_1)^\approx \,  
		t_{(2,2)} \  
		(t_{(2,2)} t_{(1,2)})^{\omega}\\
%--------------------------------------------------------------------------------
	& = & 	\pi_\approx \big( 
		\otimes(x_2,x_2)[0,h_2)^\approx,\, 
		\big(\otimes(x_2,x_2)[h_{2n},h_{2n+1})^\approx,\, 
		 \otimes(x_1,x_2)[h_{2n+1},h_{2n+2})^\approx
		\big)_{n  \in \Nat^+}
		\big)\\
%--------------------------------------------------------------------------------
	& = & 	\otimes(y_2,x_2)^\approx	
\end{eqnarray*}

Thus, if $y_1 \approx y_2$ then also $y_2 \approx x_2$ and so {\bf by transitivity}
$y_1 \approx x_2$. But since $y_1 \eqe x_1$, the $\approx$-class of $x_2$
meets the $\eqe$-class of $x_1$, contradicting the initial choice of the $x_i$s. \\

\noindent
{\em Items 2 and 3.} 
Define intermediate semigroup elements $q:= s_1$, $r := t_{(1,1)}$, $r^\uparrow
:= t_{(1,2)} t_{(1,1)}$ and $r^\downarrow := t_{(2,1)} t_{(1,1)}$.  
Then
\begin{enumerate}
\item both $r^\uparrow$ and $r^\downarrow$ absorb $r$, since $t_{(1,1)}$ is idempotent;
\item $\otimes(y_j,y_j)[0,h_2)^\approx = \otimes(y_j,y_j)[0,h_1)^\approx t_{(2,2)}$ 
      and thus absorbs $r$ (for $y_j \in \{y_1,y_2\}$). 
\end{enumerate}

\noindent
In this notation, for all $i \in \Nat^+$ and $y_j \in \{y_1,y_2\}$,
\begin{itemize}

\item $\otimes(y_j,\vec{z})[h_{2i},h_{2i+2})^\varphi$ is	 $q q = q$,

\item $\otimes(y_j,y_j)[h_{2i},h_{2i+2})^\approx$ is 	 $r r  = r$,

\item $\otimes(y_1,y_2)[h_{2i},h_{2i+2})^\approx$ is 	 $t_{(1,2)} t_{(1,1)}  =  r^\uparrow$,

\item $\otimes(y_2,y_1)[h_{2i},h_{2i+2})^\approx$ is 	 $t_{(2,1)} t_{(1,1)}  =  r^\downarrow$.\\

\end{itemize}

%Thus 
%\begin{enumerate}
%
%\item $\{h_{2i}\}_{i \in \Nat^+}$ is $\varphi,q$-homogeneous for both
%$\otimes(y_1,\vec{z})$ and $\otimes(y_2,\vec{z})$, 
%
%\item $\{h_{2i}\}_{i \in \Nat^+}$ is $\approx,r$-homogeneous for both $\otimes(y_1,y_1)$ and $\otimes(y_2,y_2)$,
%
%\item $\{h_{2i}\}_{i \in \Nat^+}$ is $\approx,r^\uparrow$-homogeneous for $\otimes(y_1,y_2)$, and
%
%\item $\{h_{2i}\}_{i \in \Nat^+}$ is $\approx,r^\downarrow$-homogeneous for $\otimes(y_2,y_1)$.
%\end{enumerate}

Finally, define the set $G:=\{h_{2ki}\}_{i > 1}$, i.e. $g_i = h_{2k(i+1)}$, 
and the semigroup elements $t:= r^k$, $t^\uparrow := (r^\uparrow)^k$, 
$t^\downarrow := (r^\downarrow)^k$ and $s:= q^k$. 
The extra multiple of $k$ (defined as the product of the exponents of the give
semigroups for $\eqe$ and $\approx$) ensures all these semigroup elements (in
particular $t^\uparrow$ and $t^\downarrow$) are idempotent.
We now verify the absorption properties:
\begin{eqnarray*}
  t^\uparrow t = r^{\uparrow k} r^k = r^{\uparrow k} = t^{\uparrow} & 
    \quad \text{because } r^{\uparrow} \text{ absorbs } r 
\end{eqnarray*}
% \begin{eqnarray*}
% t^\uparrow t 		& = & r^{\uparrow k} r^k \\
% 			& = & (t_{(1,2)} t_{(1,1)})^{k-1} t_{(1,2)} t_{(1,1)} t_{(1,1)}^k\\
% 			& = & (t_{(1,2)} t_{(1,1)})^{k-1} t_{(1,2)} t_{(1,1)}\\
% 			& = & (t_{(1,2)} t_{(1,1)})^{k-1} t_{(1,2)} t_{(1,1)}\\
% 			& = & t^\uparrow.
% \end{eqnarray*}

Similarly, $t^\downarrow t$ absorbs $t$. Further, since $g_1 = h_{4k}$, we have 
\begin{eqnarray*}
\otimes(y_j,y_j)[0,g_1)^{\approx} 
		& = & \otimes(y_j,y_j)[0,h_2)^{\approx} \otimes(y_j,y_j)[h_2,h_{4k})^\approx\\
		& = & \otimes(y_j,y_j)[0,h_2)^{\approx} r^{4k-2}\\
		& = & \otimes(y_j,y_j)[0,h_2)^{\approx} r^{3k-2}t\\
\end{eqnarray*}
and thus absorbs $t$. 

Finally we verify the homogeneity properties:
$G$ is an $\approx,t^\downarrow$-homogeneous factorisation of $\otimes(y_2,y_1)$ since for $i \in \Nat^+$
\begin{eqnarray*}
  \otimes(y_2,y_1)[g_i,g_{i+1})^\approx 	
      \ = \  \otimes(y_2,y_1)[h_{2k(i+1)},h_{2k(i+2)})^\approx  
      \ = \  (r^\downarrow)^k 
      \ = \  t^\downarrow.
\end{eqnarray*}
The other cases are similar.
\end{proof}

\vskip-0.3cm
\subsection{Shuffling step}

We continue the proof of Proposition \ref{prop_count} by 'shuffling' the
words $y_1$ and $y_2$ along $G$ resulting in continuum many pairwise distinct
%(in fact pairwise $\not \eqe$) 
words that are pairwise not $\approx$-equivalent, each satisfying
$\varphi(-,\vec{z})$.
To this end, define for $S \subset \Nat^+$ the 'characteristic word' $\chi_S$ by 
\begin{eqnarray*}
\chi_S[0,g_1) & := &y_2[0,g_1) \text{ , and }\\
\chi_S[g_n,g_{n+1}) & := &
\begin{cases} 	
		y_2[g_n,g_{n+1}) & \text{ if } n \in S \\
		y_1[g_n,g_{n+1}) & \text{ otherwise}
\end{cases}
\end{eqnarray*}

First note that $\frakA \models \varphi(\chi_S,\vec{z})$. 
Indeed, by Lemma \ref{lemma_ab} item 2
\begin{eqnarray*}
\otimes(\chi_S,\vec{z})^\varphi 	& = & \otimes(y_2,\vec{z})[0,g_1)^\varphi s^\omega\\
				& = & \otimes(y_2,\vec{z})^\varphi 
\end{eqnarray*}
and $\frakA \models \varphi(y_2,\vec{z})$ by Lemma \ref{lemma_ab} item 1. 
Moreover, for $S \not \eqe T$ the construction gives that $\chi_S \not \eqe \chi_T$.
This is due our initial choice of $x_1 \not \eqe x_2$ and the assumption that
the factorisation $(h_n)_n$ is coarse enough so that $x_1[h_n,h_{n+1}) \neq x_2[h_n,h_{n+1})$
and therefore also $y_1[g_n,g_{n+1}) \neq y_2[g_n,g_{n+1})$ for all $n$.

The following two lemmas establish that if $S \not \eqe T$ 
then $\chi_S \not \approx \chi_T$. 

Write $x_{\circ\bullet}$ for the word $\chi_{2\Nat^+}$, and $x_{\bullet\circ}$ for $\chi_{2\Nat^+ -1}$
and let $p$ denote $\otimes(y_2,y_2)[0,g_1)^\approx$.

\begin{lemma} \label{lemma_shuffle}
For all $S \not \eqe T$, 
\[
\otimes(\chi_S,\chi_T)^\approx= 
\begin{cases}
\otimes(x_{\circ\bullet},x_{\bullet\circ})^\approx & \text{ or }\\
\otimes(x_{\bullet\circ},x_{\circ\bullet})^\approx
\end{cases}
\]
\end{lemma}
\begin{proof}

Define semigroup-elements $p_n$ for $n \in \Nat$ by 
 		
\[
p_n := 
		\begin{cases}	
		t^\downarrow & \text{ if } n \in S \setminus T \\
		t^\uparrow   & \text{ if } n \in T \setminus S \\
		t	     & \text{ otherwise}
		\end{cases}
\]

Let $m$ be the smallest number in $S \triangle T$. Suppose that $m \in S \setminus T$.
Because both $t^\uparrow$ and $t^\downarrow$ are idempotent and 
since $t$ is absorbed by both $p$, $t^\uparrow$ and $t^\downarrow$ we have
\begin{eqnarray*}
\otimes(\chi_S,\chi_T)^\approx	
%--------------------------------------------------------------------------------
	& = & 	\pi_\approx \left( 
		p, \, 
		(p_n)_{n \in \Nat}
		\right) 
%--------------------------------------------------------------------------------
	\  = \ 	p (t^\downarrow t^\uparrow)^\omega \\
%--------------------------------------------------------------------------------
	& = & 	\otimes(x_{\bullet\circ},x_{\circ\bullet})^\approx
\end{eqnarray*}
and the case that $m \in T \setminus S$ similarly results 
in $\otimes(x_{\circ\bullet},x_{\bullet\circ})^\approx$.
\end{proof}

\begin{lemma}
$x_{\circ\bullet} \not \approx x_{\bullet\circ}$. 
\end{lemma}

\begin{proof}
Define an intermediate word $x_{\circ\bullet\circ\circ}:=\chi_{4\Nat^+ - 2}$.
By computations similar to the above we find that
\begin{eqnarray*}
\otimes(x_{\bullet\circ},x_{\circ\bullet\circ\circ})^\approx	
%--------------------------------------------------------------------------------
	& = & 	p (t^\downarrow t^\uparrow t^\downarrow t)^\omega 
%--------------------------------------------------------------------------------
	\ = \ 	p (t^\downarrow t^\uparrow t^\downarrow)^\omega 
%--------------------------------------------------------------------------------
	\ = \  	p (t^\downarrow t^\uparrow)^\omega \\
%--------------------------------------------------------------------------------
	& = & 	\otimes(x_{\bullet\circ},x_{\circ\bullet})^\approx 
\end{eqnarray*}
and
\begin{eqnarray*}
\otimes(x_{\circ\bullet},x_{\circ\bullet\circ\circ})^\approx	
%--------------------------------------------------------------------------------
	& = & 	p (t t t t^\downarrow)^\omega 
%--------------------------------------------------------------------------------
	\ = \  	p (t^\downarrow)^\omega \\
%--------------------------------------------------------------------------------
	& = & 	\otimes(y_2,y_1)^\approx 
\end{eqnarray*}

Therefore, if $x_{\bullet\circ} \approx x_{\circ\bullet}$ then 
also $x_{\bullet\circ} \approx x_{\circ\bullet\circ\circ}$ and so 
{\bf by symmetry} and {\bf by transitivity} 
$x_{\circ\bullet} \approx x_{\circ\bullet\circ\circ}$. 
But in this case also $y_2 \approx y_1$, contradicting Lemma \ref{lemma_ab} item 1. 
\end{proof}

\noindent There are continuum many classes in $\calP(\mathbb{N})/\sim_e$,
thus there is a continuum of pairwise not $\approx$-equivalent 
words $\chi_S$ each satisfying $\varphi(-,\vec{z})$.
This completes the proof of Proposition~\ref{prop_count}.
\end{proof}

%%%%%%%%%%%%%%%%%%%%%%%%%%%%%%%%%%%%%%%%%%%%%%%%%%%%%%%%%%%%%%%%%%%%%%%%%%%%%
%%
%% Consequences for $\omega$-automatic presentations
%%
%%%%%%%%%%%%%%%%%%%%%%%%%%%%%%%%%%%%%%%%%%%%%%%%%%%%%%%%%%%%%%%%%%%%%%%%%%%%%
\vskip-0.7cm
\section{Consequences} % for $\omega$-automatic presentations}

\noindent
{\bf Theorem \ref{thrm_foc}}
{\it The statements of Theorem \ref{thrm_fo} hold true for $\FOC$ 
     over all (not necessarily injective) $\omega$-automatic presentations.}

\begin{proof}
We prove item $(i)$ from which the rest of the theorem 
follows immediately. We inductively eliminate occurrences of cardinality
and modulo-counting quantifiers in the following way.

The countability quantifier $\exists^{\leq \aleph_0}$ and uncountability
quantifier $\exists^{>\aleph_0}$ can be eliminated (in an extension of the
presentation by $\eqe$) by the formula given in Proposition \ref{prop_count}.  

For the remaining quantifiers we further expand the presentation with the 
$\omega$-regular relations
\begin{itemize}
\item $\pi(a,b,c)$ saying that $a \eqe b \eqe c$ and the last position where $a$
differs from $c$ is no larger than the last position where $b$ differs from
$c$, and
\item $\lambda(a,b,c)$ saying that $\pi(a,b,c)$ and $\pi(b,a,c)$ and, writing
$k$ for this common position, the word $a[0,k]$ is lexicographically smaller
than the word $b[0,k]$.
\end{itemize}

Now $\exists^{<\infty} . \, \varphi(x,\vec{z})$ is equivalent to

\[
\exists x_1 \cdots x_C \, \Psi(x_1,\cdots,x_C,\vec{z})
\]
where $\Psi$ expresses that $x_1, \cdots x_C$ satisfy  $\varphi(-,\vec{z})$ and
there exists a position, say $k \in \Nat$, so that every $\approx$-class contains a
word satisfying $\varphi(-,\vec{z})$ that coincides with one of the $x_i$ from
position $k$ onwards. This additional condition can be expressed
by 
\[
\exists y_1 \cdots y_C 
  \forall x 
    \exists y 
	\left( 
           \varphi(x,\vec{z})  \limp
	       x \approx y  \land  \bigvee_i \pi(y,y_i,x_i)
	\right)
\]

Consequently, $\exists^{(r\,\mathrm{mod}\,m)} x \, . \, \varphi(x,\vec{z})$ can be
eliminated since we can pick out unique representatives of the
$\approx$-classes as those $x$ so that, writing $i(w)$ for the smallest index $i$ 
for which $w \eqe x_i$, for every $y \neq x$ in the same $\approx$-class as $x$,
either 
\begin{itemize}
\item $i(x) < i(y)$, or
\item $i(x) = i(y)$ and $\lambda(x,y,x_{i(x)})$. 
\end{itemize} 
Now we can apply the construction of \cite{KL06} or \cite{KRS04} for elimination 
of the $\exists^{(r\,\mathrm{mod}\,m)}$ quantifier.
\end{proof}

As a corollary of Proposition~\ref{prop_count} we obtain that for every 
omega-regular equivalence with countably many classes a set of unique representants 
is definable.

\begin{corollary}
Let $\approx$ be an $\omega$-automatic equivalence relation on $\Sigma^\omega$.
There is a constant $C$, depending on the presentation, so that 
the following are equivalent:
\begin{enumerate}
\item $\approx$ has countably many equivalence classes.
\item there exist $C$ many $\eqe$-classes so that every $\approx$-class 
has non-empty intersection with at least one of these $C$. 
%\item there at most $C$ many pairwise $\approx$-inequivalent words that are
%also $\eqe$-inequivalent.
\end{enumerate}
In this case there is an $\omega$-regular set of representatives of $\approx$.
Moreover an automaton for this set can be effectively found given an automaton
for $\approx$.
\end{corollary}

\begin{proof}
The first two items are simply a specialisation of Proposition \ref{prop_count}.
We get the representatives as follows.

Write $A$ for the domain of $\approx$ and 
consider the formula $\psi(x_1,\cdots,x_C)$ with free variables $x_1,\cdots,x_C$:
\[
\bigwedge_i x_i \in A \land (\forall x \in A) (\exists y)\, [x \approx y \land \bigvee_i y \eqe x_i]
\]

The relation defined by $\psi$ is $\omega$-regular since it is a first order
formula over $\omega$-regular relations. By assumption it is non-empty. Thus it
contains an ultimately periodic word of the form $\otimes(a_1,\cdots,a_C$).
Thus each of these $a_i$s is ultimately periodic; say $a_i = v_i (u_i)^\omega$.

Then every $x$ has an $\approx$-representative in $B:=\bigcup_i \Sigma^{\ast}
(u_i)^\omega$. It remains to prune $B$ to select unique representatives for
each $\approx$-class.  

It is easy to construct an $\omega$-regular well-founded linear
order on $B$.  For every $w \in B$, let $p(w) \in \Sigma^\ast$ be the
length-lexicographically smallest word such that $w$ has period $p(w)$. Also
let $t(w) \in \Sigma^\ast$ be the length-lexicographically smallest word so
that $w = t(w) \cdot p(w)^{\omega}$. Define an order $\prec$ on $B$ by $w \prec
w'$ if $p(w)$ is length-lexicographically smaller than $p(w')$, or otherwise if
$p(w) = p(w')$ and $t(w)$ is length-lexicographically smaller than $t(w')$.
The ordering $\prec$ is $\omega$-regular since it is $\FO$-definable in terms
of $\omega$-regular relations. Finally, the required set of representatives may be 
defined as the set of $\prec$-minimal elements of every $\approx$-class;
and an automaton for this set can be constructed from an automaton for 
$\approx$.
\end{proof}

This immediately yields an \emph{injective} $\omega$-automatic presentation
from a given $\omega$-automatic presentation which by Proposition
\ref{prop_cnt_inj} can be transformed into an automatic presentation of the
structure. Thus we conclude that every countable $\omega$-automatic structure
is already automatic. \\[0.5 em]

\noindent
{\bf Corollary \ref{coroll_cnt}}
{\it A countable structure is $\omega$-automatic if and only if it is automatic.
Transforming a presentation of one type into the other can be done effectively.
}\\[0.5 em]

Note that some of our technical results, in particular Lemmas \ref{lemma_ab}
and \ref{lemma_shuffle}, only require transitivity of the relation $\approx$
and do not use symmetry. Applying them to an $\omega$-automatic linear order
$\prec$ we get an interesting uncountable set of words of the form
$\chi_S, S \subseteq \Nat$.
For any two such words with $S \not \eqe T$, whether $\chi_S \prec \chi_T$ or not
depends only on the first position $m \in S \triangle T$.
Thus, $\prec$ behaves like the lexicographic order on such words.

\vskip-0.3cm
\subsection{Failure of L\"owenheim-Skolem theorem for $\omega$-automatic structures}

While so far the area of automatic structures has mainly focused on
individual structures, it is interesting to look at their theories as well.
We note a consequence of our work for 'automatic model theory'. 

An automatic version of the Downward L\"owenheim-Skolem Theorem would
say that every uncountable $\omega$-automatic structure has a countable
elementary substructure that is also $\omega$-automatic.  Unfortunately this is
false since there is a first-order theory with an $\omega$-automatic model but
no countable $\omega$-automatic model.  Indeed, consider the first-order theory
of atomless Boolean Algebras. Kuske and Lohrey \cite{KL06} have observed that
it has an uncountable $\omega$-automatic model, namely $(\calP(\Nat), \cap,
\cup, \neg)/\eqe$.  However, Khoussainov et al. \cite{KNRS04} show that the
countable atomless Boolean algebra is not automatic and so, by Corollary~\ref{coroll_cnt}, 
neither $\omega$-automatic.

%While we are still unable to prove that the former has no injective
%$\omega$-automatic presentation we observe that its theory is decidable by
%automata.

Here is the closest we can get to an automatic Downward L\"owenheim-Skolem
Theorem for $\omega$-automatic structures.

\begin{proposition} \label{prop_LS_omegaAS}
Let $(D,\approx,\{R_i\}_{i\leq\omega})$ be 
an omega-automatic presentation of $\frakA$ and let $\frakA_{\textrm{up}}$ be 
its restriction to the ultimately periodic words of $D$.
Then $\frakA_{\textrm{up}}$ is a countable elementary substructure of $\frakA$.
\end{proposition}
\begin{proof} 
Relying on the Tarski-Vaught criterion for elementary substructures 
we only need to show that for all first-order formulas $\varphi(\vec{x},y)$
and elements $\vec{b}$ of $\frakA_{\textrm{up}}$
\[ 
  \frakA \models \exists y \varphi(\vec{b},y) \quad \Rightarrow \quad 
  \frakA_{\textrm{up}} \models \exists y \varphi(\vec{b},y) \ .
\]
By Theorem \ref{thrm_fo} $\varphi(\vec{x},y)$ defines an 
omega-regular relation and, similarly, since the parameters $\vec{b}$ are 
all ultimately periodic the set defined by $\varphi(\vec{b},y)$ is omega-regular.
Therefore, if it is non-empty, then it also contains an ultimately periodic word,
which is precisely what we needed.
\end{proof}

This proof can be viewed as a model construction akin to a classical
compactness proof.  Indeed, starting with ultimately constant words and
throwing in witnesses for all existential formulas satisfied in $\frakA$ in
each round one constructs an increasing sequence of substructures comprising
ultimately periodic words of increasing period lengths. The union of these is
closed under witnesses by construction. The argument is valid for relational
structures with constants assuming that every constant is represented by an
ultimately periodic word. \\

\noindent {\bf Future work} It remains to be seen whether statements analogous 
to Theorem~\ref{thrm_foc} and Corollary~\ref{coroll_cnt} also hold for automatic 
presentations over infinite trees. \\

\noindent {\bf Acknowledgment} We thank the referees for detailed technical 
remarks and corrections.

% One intriguing long standing question is whether the group of rationals
% $(\mathbb{Q},+)$ is automatic or not.  This is of course the structure we
% obtain when restricting the natural $\omega$-automatic presentation of
% $(\mathbb{R},+)$ to its ultimately periodic elements.  Should $(\mathbb{Q},+)$
% prove not to be automatic it would provide another $\FO$-theory with an
% $\omega$-automatic model but no countable $\omega$-automatic model. 
% % Either way,
% % we are pleased to note that we have an automata-theoretic proof of the fact
% % that $(\mathbb{Q},+)$ has decidable $\FO$ (in fact $\FOC$) theory.

%%%%%%%%%%%%%%%%%%%%%%%%%%%%%%%%%%%%%%%%%%%%%%%%%%%%%%%%%%%%%%%%%%%%%%%%%%%%%
%%
%% Bibliography
%%
%%%%%%%%%%%%%%%%%%%%%%%%%%%%%%%%%%%%%%%%%%%%%%%%%%%%%%%%%%%%%%%%%%%%%%%%%%%%%


\vskip-0.3cm
\begin{thebibliography}{10}

\bibitem{Blu99}
A.~Blumensath.
\newblock Automatic structures.
\newblock {Diploma thesis, RWTH-Aachen}, 1999.

\bibitem{BG00}
A.~Blumensath and E.~Gr{\"a}del.
\newblock {Automatic Structures}.
\newblock In {\em Proceedings of 15th IEEE Symposium on Logic in Computer
  Science LICS 2000}, pages 51--62, 2000.

\bibitem{BG04}
A.~Blumensath and E.~Gr{\"a}del.
\newblock Finite presentations of infinite structures: Automata and
  interpretations.
\newblock {\em Theory of Comp. Sys.}, 37:641 -- 674, 2004.

\bibitem{HKMNman}
G.~Hj{\"o}rth, B.~Khoussainov, A.~Montalban, and A.~Nies.
\newblock Borel structures.
\newblock Manuscript, 2007.

\bibitem{Hod83}
B.R. Hodgson.
\newblock D\'ecidabilit\'e par automate fini.
\newblock {\em Ann. sc. math. Qu\'ebec}, 7(1):39--57, 1983.

\bibitem{KN95}
B.~Khoussainov and A.~Nerode.
\newblock Automatic presentations of structures.
\newblock In {\em LCC '94}, volume 960 of {\em LNCS}, pages 367--392.
  Springer-Verlag, 1995.

\bibitem{KNRS04}
B.~Khoussainov, A.~Nies, S.~Rubin, and F.~Stephan.
\newblock Automatic structures: Richness and limitations.
\newblock In {\em LICS}, pages 44--53. IEEE Comp. Soc., 2004.

\bibitem{KRS04}
B.~Khoussainov, S.~Rubin, and F.~Stephan.
\newblock Definability and regularity in automatic structures.
\newblock In {\em STACS '04}, volume 2996 of {\em LNCS}, pages 440--451, 2004.

\bibitem{KL06}
D.~Kuske and M.~Lohrey.
\newblock First-order and counting theories of $\omega$-automatic structures.
\newblock In {\em FoSSaCS}, pages 322--336, 2006.

\bibitem{PP95}
D.~Perrin and J.-E. Pin.
\newblock Semigroups and automata on infinite words.
\newblock In J.~Fountain, editor, {\em Semigroups, Formal Languages and
  Groups}, NATO Advanced Study Institute, pages 49--72. Kluwer, 1995.

\end{thebibliography}
\end{document}